\documentclass[journal]{IEEEtran}%[11pt,draftcls,onecolumn]{IEEEtran}%[journal]{IEEEtran}%[12pt,draftcls,onecolumn]
\usepackage{amssymb}
\usepackage{amsmath}
\usepackage{cite}
\usepackage{algorithm}
\usepackage{algorithmic}
\usepackage[dvips]{graphicx}
\usepackage{stfloats}
\begin{document}
%
% paper title
\title{Regularized Zero-Forcing for Multiantenna Broadcast Channels with User Selection}

\author{Zijian Wang, and Wen Chen,~\IEEEmembership{Senior Member,~IEEE}
\thanks{Manuscript received November 28th, 2011; revised January 11th, 2012, accepted
February 8th, 2012. The associate editor coordinating the review of
this paper and approving it for publication was D. Huang.}
\thanks{The authors are with
Department of Electronic Engineering, Shanghai Jiao Tong University,
Shanghai, and SKL for ISN, Xidian University, China. e-mail:
\{yuyang83;wenchen\}@sjtu.edu.cn.}
\thanks{This work is supported by NSFC \#60972031, by national 973 project
\#2012CB316106 and \#2009CB824900, by NSFC \#61161130529, by
national key laboratory project \#ISN11-01.}
%F.
%Gao is with School of Engineering and Science, Jacobs University,
%Bremen, Germany, 28759. Email: feifeigao@ieee.org. C. Tellambura is with the Department of Electrical and Computer
%Engineering, University of Alberta, Canada. Email: chintha@ece.ualberta.ca.
}
%\thanks{This work is supported by NSF China \#60972031, by SEU SKL project
%\#W200907, by ISN project \#ISN11-01, by Huawei Funding
%\#YJCB2009024WL and \#YJCB2008048WL,
%and by National 973 project \#2009CB824900.}% <-this % stops a space

% The paper headers
%\markboth{IEEE Transactions on Vehicular Technology.}{}
\maketitle

\begin{abstract}
A multiantenna multiuser broadcast channel with transmitter beamforming and user selection is considered. Different from the conventional works, we consider imperfect channel state information (CSI) which is a practical scenario for multiuser broadcast channels. We propose a robust regularized zero-forcing (RRZF) beamforming at the base station. Then we show that the RRZF outperforms zero-forcing (ZF) and regularized ZF (RZF) beamforming even as the number of users grows to infinity. Simulation results validate the advantage of the proposed robust RZF beamforming.
\end{abstract}


%\newpage
\section{Introduction}
In recent years, multiple-input multiple-output (MIMO) has drawn considerable interest due to the advantages of increasing the data rate~\cite{MIMO}. Several beamformings have been presented in the literature to provide the multiplexing gain. But for multiantenna broadcast channels, only precodings can be implemented at the transmitter because the receivers do not mutually cooperate. Linear transmit precodings for broadcast channels have been studied in~\cite{ZF,swindlehurst}.

For broadcast channels with large number of users,  user selection
is necessary to provide multiuser diversity. In \cite{Yoo,Yoo2}, the
authors propose zero-forcing (ZF) beamforming at the transmitter in
conjunction with a semiorthogonal user selection (SUS) algorithm.
Performance analysis of ZF beamforming is studied in
\cite{victoria}. In \cite{hassibi}, different beamforming and user
selection schemes are compared and analyzed. To deal with the poor
performance of ZF for small number of users, beamformings based on
hybrid zero-forcing and orthogonal beamforming \cite{WeiXu} and
channel inversion regularization \cite{RZFselect} are proposed.
Methods to reduce the feedback needed for user selection have been
studied in \cite{reporting}.

In this letter, we propose a robust regularized  zero-forcing (RRZF)
beamforming, where the user selection is based on the SUS algorithm
as in \cite{Yoo}. While the ZF beamforming and regularized ZF (RZF)
have degraded performance for imperfect channel state information
(ICSI), the proposed RRZF significantly improves the performance. While the
conventional optimal $\alpha$ in \cite{RZFselect,wan} is $M/\rho$, where
$M$ is the number of transmit antenna and $\rho$ is the
signal-to-noise-ratio (SNR), we found that the optimal $\alpha$
grows with the number of users.
Although the RRZF is optimized for small number of users, we show
that in the extremal case when the number of users is infinity, the
sum rate performance of RRZF still outperforms the ZF and RZF
beamforming.  Especially, we show that in this extremal case, the
sum rate is monotonically  increasing with the regularizing factor
$\alpha$, and the optimal $\alpha$ is infinity.

In this letter, boldface lowercase letter and boldface uppercase
letter represent vectors and matrices, respectively.  Notations
$\|\mathbf{a}\|$  stands for the Euclidean norm of a vector
$\mathbf{a}$ and $|{a}|$ stands for the modulus of a complex ${a}$
respectively. $\mathrm{tr}(\cdot)$ and $(\cdot)^H$ denote the trace
and conjugate transpose operation of a matrix. Term $\mbox{\boldmath
$\mathbf{I}$}_N$ is an $N{\times}N$ identity matrix.
$\overset{w.p.}{\longrightarrow}$ represents convergence with
probability one. Finally, we denote the expectation operation by
$\mathrm{E}\left\{\cdot\right\}$.

\section{System Model}
We consider a multiantenna multiuser broadcast network which consists of a base
station equipped with $M$ antennas, and $K$ user terminals each with only a single
antenna. It is assumed that $K>M$. So the base station needs to choose $M$ favorable users out of the $K$ users to transmit $M$ datas simultaneously. Then the base station broadcasts $M$ precoded data streams after applying a linear precoder to the original data vector
${\bf{s}}\in \mathbb{C}^M$, where $\mathrm{E} \{{\bf{s}}{\bf{s}}^H
\}={\bf{I}}_M$. We denote the precoding matrix at the base station
as ${\bf{W}}$ and suppose that the base station transmit power is
$P$.  A power control factor can be derived as
\begin{equation}\label{6}
\rho=\sqrt{\frac{P}{\mathrm{E}\{{\bf{s}}^H {\bf{W}}^H{\bf{W}} {\bf{s}}\}}}=\sqrt{\frac{P}{\mathrm{tr}({\bf{W}}^H {\bf{W}})}}.
\end{equation}
The received signal vector at the selected $M$ user terminals is
\begin{equation}\label{2}
{\bf{y}}=\rho {\bf{H}} {\bf{W}} {\bf{s}}+{\bf{n}},
\end{equation}
where ${\bf{H}}\in  \mathbb{C}^{M\times M}$ is the Rayleigh broadcast channel
matrix from the base station to the $M$ selected users, in which, all entries are  $i.i.d$ complex Gaussian
distributed with zero mean and unit variance, and ${\bf{n}}\in
\mathbb{C}^M$  is the noise vector, in which, all the
entries are $i.i.d$ complex Gaussian distributed with zero mean and
variance $\sigma^2$.

From (\ref{2}), the received signal at the $k$-th user can be rewritten as
\begin{equation}\label{5}
{{y}}_k=\rho {\bf{h}}_k^H {\bf{W}} {\bf{s}}+{{n}}_k
=\rho {\bf{h}}_k^H {\bf{w}}_k {s}_k+\sum_{j=1,j\neq k}^M\rho {\bf{h}}_k^H{\bf{w}}_j{s}_j+{{n}}_k,
\end{equation}
where ${\bf{w}}_k$ is the $k$-th column of ${\bf{W}}$ and ${\bf{h}}_k^H$ is the $k$-th row of $\mathbf{H}$ denoting the channel vector from the base station to the $k$-th user.
Therefore, the signal-to-interference-plus-noise ratio (SINR) of the $k$-th user is
\begin{equation}\label{4}
\mathrm{SINR}_k=\frac{\rho^2 |{\bf{h}}_k^H {\bf{w}}_k|^2}{\rho^2 \sum_{j=1,j\neq k}^M |{\bf{h}}_k^H{\bf{w}}_j|^2+\sigma^2}.
\end{equation}

For we aim to analyze the RZF for user selection instead of finding the optimal algorithm, we generalize a simplified SUS (semiorthogonal user selection) algorithm in \cite{Yoo} as follows. It will be stopped when $|\mathcal{S}|=M$.

Step 1) Initialization:
\begin{eqnarray}
\mathcal{X}_1=\left\{1,\ldots,K\right\};\quad
i=1;\quad
\mathcal{S}=\phi;
\end{eqnarray}

Step 2) Select the $i$th user as follows:
\begin{eqnarray}
\pi(i)=\mathrm{arg} \underset{k\in\mathcal{X}_{i}}\max\|\mathbf{h}_k\|;\quad
S\leftarrow S\cup {\pi(i)};
%\mathbf{h}_{(i)}=\mathbf{h}_{\pi(i)};
\end{eqnarray}

Step 3) If $|\mathcal{S}|<M$, then calculate $\mathcal{X}_{i+1}$, and the set of users semiorthogonal to $\mathbf{h}_{\pi(i)}$
\begin{eqnarray}
\mathcal{X}_{i+1}&=&\left\{k\in \mathcal{X}_i, k \neq \pi(i) | \frac{|{\bf{h}}_{\pi(i)}^H {\bf{h}}_k|}{\|{\bf{h}}_{\pi(i)}\|\|{\bf{h}}_k\|}<\beta\right\};\label{1}\\
i&\leftarrow&i+1.
\end{eqnarray}
In every step, the algorithm selects the best user among the user pool which are semiorthogonal to the selected users.
\section{RRZF for ICSI and performance analysis}
%In~\cite{Yoo}, ZF beamforming is used at the base station. But the ZF beamforming has two disadvantages. First, it has poor sum rate performance when the number of users ($K$) is small or when in low SNR regime. Second, the performance degrades rapidly when imperfect CSI occurs.
In this section, we first propose an RRZF beamforming at the base station considering ICSI. The  regularizing factor $\alpha$ in RRZF is larger than that in RZF since additional noise  inherited from the CSI error is considered. Then we show that in the extremal case where the number of users is infinity, the sum rate is monotonically increasing with the $\alpha$, which implies that the proposed RRZF outperforms ZF and RZF. Since it is difficult to obtain the distribution of channel matrix for moderate user number, we give simulation results of optimal $\alpha$ in Fig.~1.
\subsection{RRZF beamforming for ICSI}
The power penalty problem exists in ZF  because the beamforming
vector does not match with the channel vector for each user. This
can be solved by selecting users with nearly orthogonal channel
vectors.
%Consider an extremal case $\beta=0$ in (\ref{1}). Then the
%ZF beamforming becomes
%\begin{equation}
%\mathbf{W}_{ZF}=\left[\frac{\mathbf{h}_1}{\|\mathbf{h}_1\|^2},\ldots,\frac{\mathbf{h}_M}{\|\mathbf{h}_M\|^2}\right],
%\end{equation}
%which matches with the channel matrix.

But it is still a severe problem for small user numbers because finding $M$ semiorthogonal users is not guaranteed. Adding an identity matrix multiplied by a regularizing factor $\alpha$ before the inversion manipulation is another efficient way to solve the power penalty problem~\cite{swindlehurst}.
Implementing RZF beamforming, we have
$
{\bf{W}}=\mathbf{H}^H\left(\mathbf{H}\mathbf{H}^H+\alpha\mathbf{I}\right)^{-1}
$
in (\ref{5}).
Note that the channel inversion regularization brings interference among different users if $\alpha\neq 0$. The optimal tradeoff of $\alpha$ is obtained in~\cite{swindlehurst} as
$
\alpha^{\mathrm{RZF}}=M\sigma^2/P.
$

The CSI in the practical scenario is imperfect due to large delay caused by user selection. We propose a robust RZF (RRZF) by optimizing the $\alpha$. We model the imperfect CSI as \cite{love}
\begin{equation}\label{9}
{\bf{H}}=\hat{\bf{H}}+e{\bf{\Omega}},
\end{equation}
where $e{\bf{\Omega}}$ is the CSI error independent of $\hat{\bf{H}}$, and ${\bf{\Omega}}$ is unknown to the base station and the user terminals. The entries of ${\bf{\Omega}}$ are $i.i.d$ complex Gaussian distributed with zero mean and unit
variance, and $e^2$ denotes the power of the CSI error which is known to the base station. Then the received signal vector can be rewritten as
\begin{equation}
{\bf{y}}=\hat{\rho} {\bf{H}} \hat{{\bf{W}}} {\bf{s}}+{\bf{n}}=
\hat{\rho} \hat{{\bf{H}}} \hat{{\bf{W}}} {\bf{s}}+e\hat{\rho} {\bf{\Omega}} \hat{{\bf{W}}} {\bf{s}}+{\bf{n}},
\end{equation}
where $\hat{\bf{W}}=\hat{\mathbf{H}}^H\left(\hat{\mathbf{H}}\hat{\mathbf{H}}^H+\alpha\mathbf{I}\right)^{-1}$ and $\hat{\rho}$ is derived by substituting $\hat{\bf{W}}$ into (\ref{6}).
The covariance of the noise becomes
\begin{equation}
\begin{split}
&\mathrm{E}\left\{\left(e\hat{\rho} {\bf{\Omega}} \hat{{\bf{W}}} {\bf{s}}+\mathbf{n}\right)\left(e\hat{\rho} {\bf{\Omega}} \hat{{\bf{W}}} {\bf{s}}+\mathbf{n}\right)^H\right\}\\
=&e^2{\hat{\rho}}^2\mathrm{E}\left\{{\bf{\Omega}} \hat{{\bf{W}}}
{\bf{s}} {\bf{s}}^H
\hat{{\bf{W}}}^H{\bf{\Omega}}^H\right\}+\mathrm{E}\left\{\mathbf{n}\mathbf{n}^H\right\}
=\left(e^2P+\sigma^2\right)\mathbf{I}_M,
\end{split}
\end{equation}
where we used the fact $\mathrm{E}\{\mathbf{\Omega} \mathbf{A} \mathbf{\Omega}^H\}=\mathrm{tr}(\mathbf{A})\mathbf{I}_N$ for any $N\times N$ matrix $\mathbf{A}$ \cite{trace}.
We use the diagonal decomposition
\begin{equation}\label{10}
\hat{{\bf{H}}}\hat{{\bf{H}}}^H=\mathbf{Q}\mathbf{\Lambda} \mathbf{Q}^H
\end{equation}
in the following analysis where $\mathbf{\Lambda}=\mathrm{diag}\{\lambda_1,\ldots,\lambda_M\}$ is a diagonal matrix. From (\ref{9}), the imperfect CSI is a scaled version of Rayleigh channel matrix with eigenvalues scaled by $\left(1-e^2\right)^{\frac{1}{2}}$. Since in the decomposition (\ref{10}), $\mathbf{Q}$ and $\mathbf{\Lambda}$ are independent \cite{independent}, the statistic distribution is the same as in the perfect channel matrix. Therefore, we can use the method  as in \cite{swindlehurst} of taking expectations over $\mathbf{Q}$ to the desired signal and the interference to divide the desired signal and the interference in $\hat{\rho} \hat{{\bf{H}}} \hat{{\bf{W}}} {\bf{s}}$ and finally obtain the average SINR at each user terminal as a function of the eigenvalues of $\hat{{\bf{H}}}$, that is
\begin{equation}\label{11}
\begin{split}
&\mathrm{SINR}\left(\{\lambda\}\right)\\
&=\frac{\mathrm{E}\left\{\left(\hat{\rho} \hat{{\bf{H}}} \hat{{\bf{W}}}\right)_{k,k}\right\}}{\sum_{j=1,j\neq k}^M\mathrm{E}\left\{\left(\hat{\rho} \hat{{\bf{H}}} \hat{{\bf{W}}}\right)_{k,j}\right\}+e^2P+\sigma^2}\\
&=\frac{\left(\sum\frac{\lambda}{\lambda+\alpha}\right)^2+\sum\frac{\lambda^2}{(\lambda+\alpha)^2}}{\left(e^2+\frac{\sigma^2}{P}\right)M(M+1)\sum\frac{\lambda}{(\lambda+\alpha)^2}+M\sum\frac{\lambda^2}{(\lambda+\alpha)^2}-\left(\sum\frac{\lambda}{\lambda+\alpha}\right)^2},
\end{split}
\end{equation}
where the summation $\sum$ is taken from $\lambda_1$ to $\lambda_M$. The optimal $\alpha$ can be obtained by taking derivative to (\ref{11}) and setting it to zero. After some manipulations, we have
\begin{equation}
\sum_{k<l}\frac{\lambda_k\lambda_l\left(\lambda_k-\lambda_l\right)^2\left(M\left(\frac{\sigma^2}{P}+e^2\right)-\alpha\right)}{(\lambda_k+\alpha)^3(\lambda_l+\alpha)^3}=0,
\end{equation}
which implies $\alpha^{\mathrm{RRZF}}=M\left(\frac{\sigma^2}{P}+e^2\right)$.
%Then we have
%\begin{equation}
%\hat{{\bf{W}}}=\mathbf{Q}\mathrm{diag}\{\frac{\sqrt{\lambda_1}}{\lambda_1+\alpha},\ldots,
%\frac{\sqrt{\lambda_M}}{\lambda_M+\alpha}\} \mathbf{Q}^H,
%\end{equation}
%and
%\begin{equation}
%\begin{split}
%\hat{\rho}&=\left(P\biggl/\mathrm{trace}\left(\hat{{\bf{W}}}\hat{{\bf{W}}}^H \right) \right)^{\frac{1}{2}}\\
%&=\left(\frac{1}{P}\sum_{l=1}^M\frac{{\lambda_l}}{(\lambda_l+\alpha)^2}\right)^{\frac{1}{2}}.
%\end{split}
%\end{equation}

\subsection{Performance analysis for large $K$}
In the following, we analyze the behavior of the RZF beamforming for large number of users. Imperfect CSI is assumed in the analysis. However, the conclusion also holds for perfect CSI which is a special case with $e=0$.

In the SUS algorithm,  if the $\beta$ in (\ref{1}) is too large, the selected users are not semiorthogonal enough. If it is too small, there is less user pool so that the multiuser gain is not provided. We will use the optimal $\beta$ for each $K$ in the simulations. As $K$ grows to infinity, the optimal $\beta$ decreases to zero. For an extremal case $\beta=0$, we obtain the following theorem which shows that, unlike the characteristic that RZF converges to ZF as $P/\sigma^2\rightarrow +\infty$, the RZF does not converge to ZF as $K\rightarrow +\infty$, and the proposed RRZF outperforms RZF and ZF. Note that the MF beamforming is $\mathbf{W}=\widehat{\mathbf{H}}^H$. It can be viewed as an RZF beamforming with $\alpha=+\infty$, because in this case
\begin{equation}
\begin{split}
&\rho\mathbf{W}=\sqrt{\frac{P}{\mathrm{tr}\left(\widehat{\mathbf{H}}\widehat{\mathbf{H}}^H\left(\widehat{\mathbf{H}}\widehat{\mathbf{H}}^H+\alpha\mathbf{I}_M\right)^{-2}\right)}}\widehat{\mathbf{H}}^H\left(\widehat{\mathbf{H}}\widehat{\mathbf{H}}^H+\alpha\mathbf{I}_M\right)^{-1}
\\&\overset{w.p.}\longrightarrow \sqrt{\frac{P}{\mathrm{tr}\left(\widehat{\mathbf{H}}\widehat{\mathbf{H}}^H\left(\alpha\mathbf{I}_M\right)^{-2}\right)}}\widehat{\mathbf{H}}^H\left(\alpha\mathbf{I}_M\right)^{-1}
=\sqrt{\frac{P}{\mathrm{tr}\left(\widehat{\mathbf{H}}\widehat{\mathbf{H}}^H\right)}}\widehat{\mathbf{H}}^H.
\end{split}
\end{equation}
\newtheorem{theorem}{Theorem}
\begin{theorem}
If $\beta=0$, then
\begin{equation}
\mathrm{SNR}^{\mathrm{ZF}}<\mathrm{SNR}^{\mathrm{RZF}}<\mathrm{SNR}^{\mathrm{RRZF}}<\mathrm{SNR}^{\mathrm{MF}}.
\end{equation}
\end{theorem}
\begin{proof}
When $\beta=0$, $\mathbf{h}_i^H\mathbf{h}_j=0$ for any $i\neq j$. Therefore,
\begin{equation}
\begin{split}
\widehat{\mathbf{H}}\widehat{\mathbf{H}}^H&=[\hat{\mathbf{h}_1},\ldots,\hat{\mathbf{h}_M}]^H\cdot[\hat{\mathbf{h}_1},\ldots,\hat{\mathbf{h}_M}]\\
&=\mathrm{diag}\left\{\|\hat{\mathbf{h}_1}\|^2,\ldots,\|\hat{\mathbf{h}_M}\|^2\right\}\\
&\triangleq\mathrm{diag}\left\{\lambda_1,\ldots,\lambda_M\right\}.
\end{split}
\end{equation}
Define the effective channel matrix $\mathbf{H}_{\mathrm{eff}}=\widehat{\mathbf{H}}\mathbf{W}$. We have the average SNR of each user of the RZF beamforming as
\begin{equation}
\begin{split}
\overline{\mathrm{SNR}}&=\frac{1}{M}\sum_{i=1}^M\frac{\rho^2|\left(\mathbf{H}_{\mathrm{eff}}\right)_{i,i}|^2}{\rho^2\sum_{j=1,j\neq i}^M|\left(\mathbf{H}_{\mathrm{eff}}\right)_{i,j}|^2+\left(e^2P+\sigma^2\right)}\\
&=\frac{\rho^2\mathrm{tr}\left(\mathbf{H}_{\mathrm{eff}}^2\right)}{M\left(e^2P+\sigma^2\right)}
=\frac{P\mathrm{tr}\left(\mathbf{H}_{\mathrm{eff}}^2\right)}{M\left(e^2P+\sigma^2\right)\mathrm{tr}\left(\mathbf{W}\mathbf{W}^H\right)}\\
&=\frac{P}{M\left(e^2P+\sigma^2\right)}\frac{\mathrm{tr}\left(\left(\widehat{\mathbf{H}}\widehat{\mathbf{H}}^H\left(\widehat{\mathbf{H}}\widehat{\mathbf{H}}^H+\alpha\mathbf{I}_M\right)^{-1}\right)^2\right)}{\mathrm{tr}\left(\widehat{\mathbf{H}}\widehat{\mathbf{H}}^H\left(\widehat{\mathbf{H}}\widehat{\mathbf{H}}^H+\alpha\mathbf{I}_M\right)^{-2}\right)}\\
&=\frac{P}{M\left(e^2P+\sigma^2\right)}\frac{\sum_{m=1}^M \frac{\lambda_m^2}{(\lambda_m+\alpha)^2}}{\sum_{m=1}^M\frac{\lambda_m}{(\lambda_m+\alpha)^2}}.\label{7}
\end{split}
\end{equation}
Taking derivative to (\ref{7}) with respect to $\alpha$, we have
\begin{multline}
\frac{d}{d\alpha}\frac{\sum_{m=1}^M \frac{\lambda_m^2}{(\lambda_m+\alpha)^2}}{\sum_{m=1}^M\frac{\lambda_m}{(\lambda_m+\alpha)^2}}\\
=\frac{2}{\left(\sum_{m=1}^M\frac{\lambda_m}{(\lambda_m+\alpha)^2}\right)^2}
\left(\sum_{m=1}^M\frac{\lambda_m}{(\lambda_m+\alpha)^3}\sum_{m=1}^M \frac{\lambda_m^2}{(\lambda_m+\alpha)^2}\right.\\\left.-\sum_{m=1}^M\frac{\lambda_m^2}{(\lambda_m+\alpha)^3}\sum_{m=1}^M \frac{\lambda_m}{(\lambda_m+\alpha)^2}\right),
\label{8}
\end{multline}
%The numerator in (\ref{8}) can be expanded as
where
\begin{equation}
\begin{split}
&\sum_{m=1}^M\frac{\lambda_m}{(\lambda_m+\alpha)^3}\sum_{m=1}^M \frac{\lambda_m^2}{(\lambda_m+\alpha)^2}\\
&\quad\quad\quad\quad-\sum_{m=1}^M\frac{\lambda_m^2}{(\lambda_m+\alpha)^3}\sum_{m=1}^M \frac{\lambda_m}{(\lambda_m+\alpha)^2}\\
=&\sum_{m=1}^M\frac{\lambda_m}{(\lambda_m+\alpha)^3}\sum_{m=1}^M \frac{\lambda_m^2(\lambda_m+\alpha)}{(\lambda_m+\alpha)^3}\\
&\quad\quad\quad\quad-\sum_{m=1}^M\frac{\lambda_m^2}{(\lambda_m+\alpha)^3}\sum_{m=1}^M \frac{\lambda_m(\lambda_m+\alpha)}{(\lambda_m+\alpha)^3}\\
=&\sum_{i\neq j}\frac{\lambda_i\lambda_j^2(\lambda_j+\alpha)-\lambda_i^2\lambda_j(\lambda_j+\alpha)}{(\lambda_i+\alpha)^3(\lambda_j+\alpha)^3}\\
=&\sum_{i>j}\frac{\lambda_i\lambda_j\left(\lambda_i-\lambda_j\right)^2}{(\lambda_i+\alpha)^3(\lambda_j+\alpha)^3}>0.
\end{split}
\end{equation}
Therefore, the SNR is monotonically increasing with $\alpha$. When $\alpha=0$, the beamforming is ZF. When $\alpha=+\infty$, it is MF.
\end{proof}

Therefore, the sum rate is also monotonically increasing with $\alpha$ for large number of users.
Actually, $\beta=0$ only when $K=\infty$. Therefore, as $K$ grows, although the sum rate performance of ZF improves by selecting semiorthogonal users, it remains inferior to the RZF.
From Theorem 1, we also see that for $K=+\infty$, the optimal $\alpha$ becomes $+\infty$. In fact, the conventional $\alpha^{\mathrm{opt}}=M\sigma/P$ only holds when $K=M$ because the distribution of the broadcast channel matrix $\mathbf{H}$ has changed when semiorthogonal channels are selected. The $\alpha^{\mathrm{opt}}$ grows with $K$, which is validated by simulation in Fig.1. In Fig.1, we simulate the optimal $\alpha$ versus the decreasing $\beta$ because the $\beta$ decreases as $K$ increases. We observe that $\alpha^{\mathrm{opt}}$ grows rapidly after $\beta<0.3$.
\section{Simulation Results}

In this section,  numerical results are carried out to show the
advantage of the proposed RRZF beamforming with SUS algorithm. The
performance is compared with ZF beamforming and the conventional RZF
with SUS algorithm in terms of sum rate. Both are assumed uniform
power allocation with a power control factor. For each $M$ and $K$,
we use the optimal $\beta$.

Fig.~2 shows the sum rates versus the number of users ($K$)  for low
to moderate $K$. We set $M=2,4,6$ and $P/\sigma^2=15dB$. We see that
for small $K$, the proposed robust RZF has an apparent advantage to
the conventional RZF and ZF as the $\mathrm{SINR}^{\mathrm{RRZF}}$
better balances the additional noise inherited from the CSI error.
As $K$ increases, the performance gap decreases because the power
penalty problem is solved by selecting semiorthogonal user channels.
Note that for $K=M$, the network is equivalent to a conventional
broadcast channel. In this case, as $M$ increases, the power penalty
in ZF beamforming becomes more apparent so the performance gap
between ZF and RZF grows. Note that the sum rate of both
beamformings grows like $M\log\log K$~\cite{Yoo2}.
%Fig.~3 shows the sum rates versus the number of users for large $K$. We see that the RZF still has a small performance advantage to ZF, which is consist with the analysis in Section \uppercase \expandafter {\romannumeral 3}.
\begin{figure}[!t]
\centering
\includegraphics[width=3.4in]{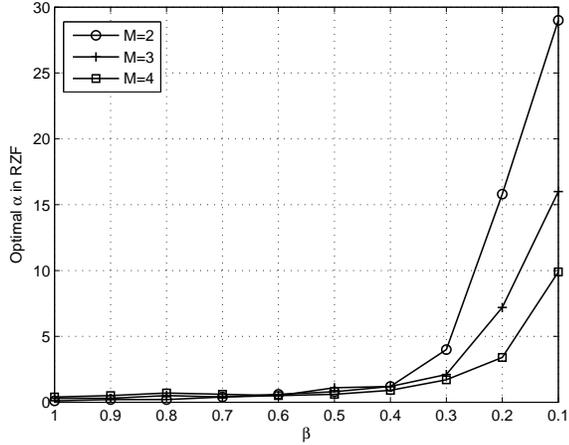}
\caption{Optimal $\alpha$ vs. $\beta$ for $M=2,3,4$, $e^2=0.1$ and
$P/\sigma^2=30dB$.}
\end{figure}
\begin{figure}[!t]
\centering
\includegraphics[width=3.4in]{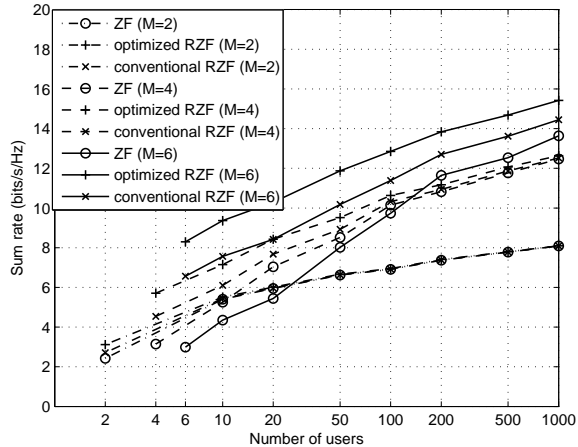}
\caption{Sum rate performances vs. the number of users. $M=2,4,6$,
$e^2=0.1$, and $P/\sigma^2=15dB$.}
\end{figure}
%\begin{figure}[!t]
%\centering
%\includegraphics[width=3.5in]{large_K_m_3_m_4.eps}
%\caption{Sum rate performances vs. the number of users for large $K$. $M=3$ and $P/\sigma^2=5dB$.}
%\end{figure}
\begin{figure}[t]
\centering
\includegraphics[width=3.4in]{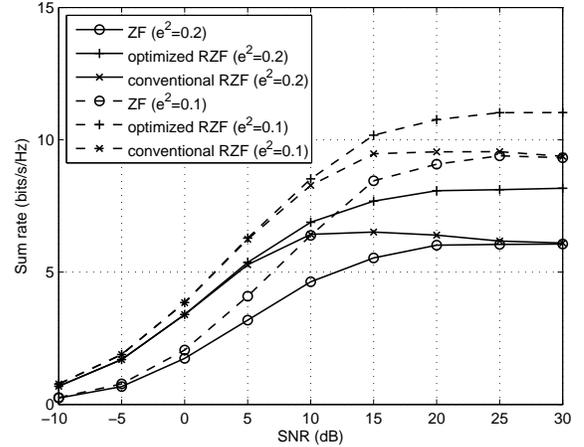}
\caption{Sum rate performances vs. the SNR of the broadcast channel.
$M=4$ and $K=20$. $e^2=0.2$ and $0.1$.}
\end{figure}

In Fig.~3, we compare the sum rates versus the power of CSI error. When CSI is imperfect, the sum rates have "ceiling effect" because the power of the desired signal and the power of the noise inherited from CSI error both goes to infinity with the SNR. The robust RZF uses $\alpha=M\left(\frac{\sigma^2}{P}+e^2\right)$ to compensate the noise and CSI error. We see that the conventional RZF converges to ZF because $\alpha=\frac{M\sigma^2}{P}\rightarrow 0$ as $P\rightarrow +\infty$. So the proposed RRZF is more robust to ZF and RZF for multiuser selection at high SNR, although it has the same performance as RZF in low SNR because the CSI error is not critical in this case.

%\begin{figure}[ht]
%\centering
%\includegraphics[width=3.5in]{_vs_e.eps}
%\caption{Sum rate performances vs. the power of CSI error. $M=4$, $K=10$ and $P/\sigma^2=20dB$.}
%\end{figure}
\section{Conclusion}
In this letter, we propose an RRZF beamforming for the multiantenna broadcast channel with the semiorthogonal user selection (SUS) algorithm for imperfect CSI. The RRZF has significant advantage to ZF and RZF for small number of users. We also show that RRZF  outperforms ZF and RZF in the extremal case of $K=+\infty$. The optimal regularizing factor $\alpha$ in RZF is no more the conventional, but increases with $K$. Since it is difficult to derive the closed-form of $\alpha$ for moderate $K$, we obtain it by monte-carlo simulations.

\end{document}